\documentclass[draft, onecolumn, letterpaper, romanappendices]{IEEEtran}
\usepackage{graphics}
\usepackage{cite}
\usepackage[pdftex]{graphicx}
\usepackage{epstopdf}
\usepackage{epsfig}
\usepackage{latexsym}
\usepackage{amsfonts}
\usepackage{calc}
\usepackage{url}
\usepackage{enumerate}
\usepackage{color}
\usepackage[tbtags]{amsmath}
\usepackage{amssymb}
\usepackage{upref}
\usepackage{dsfont}
\usepackage{multirow}
\usepackage{booktabs}
\usepackage{bigstrut}
\usepackage{rotating}
\usepackage{nth}

\newtheorem{theorem}{Theorem}

\newtheorem{lemma}{Lemma}

\newtheorem{proof}[theorem]{Proof}

\begin{document}
\title{How Fast Can Dense Codes Achieve the \\ Min-Cut Capacity of Line Networks?${}^{\dagger}$\thanks{$^{\dagger}$This paper is an extended version of a manuscript which has been submitted to IEEE ISIT 2012.}}

\author{\IEEEauthorblockN{Anoosheh~Heidarzadeh and Amir H. Banihashemi\\}
\IEEEauthorblockA{\small{Department of Systems and Computer Engineering, Carleton University, Ottawa, ON, Canada}\\
}}

\maketitle
\thispagestyle{empty}

\begin{abstract}In this paper, we study the coding delay and the average coding delay of random linear network codes (dense codes) over line networks with deterministic regular and Poisson transmission schedules. We consider both lossless networks and networks with Bernoulli losses. The upper bounds derived in this paper, which are in some cases more general, and in some other cases tighter, than the existing bounds, provide a more clear picture of the speed of convergence of dense codes to the min-cut capacity of line networks.\end{abstract}

\section{Introduction}
Random linear network codes (dense codes) achieve the capacity over various network scenarios, in particular, unicast over line networks. Lun \emph{et al.} \cite{LMKE:2008} showed that dense codes achieve the capacity of networks with transmission and loss schedules specified by stochastic processes with bounded average rate. They however did not discuss the speed of convergence of such codes to the capacity.

The speed of convergence of dense codes to the capacity of networks with arbitrary deterministic transmission schedules was studied in~\cite{MHL:2006} and~\cite{HBJ:2011}. It is not, however, straightforward to apply the results to the networks with probabilistic schedules.

In the literature, the coding delay or the average coding delay is often used to measure the speed of convergence of a code to the capacity of a network. The \emph{coding delay} of a code over a network with a given traffic (with a given schedule of transmissions and losses) is the minimum time that the code takes to transmit all the message vectors from the source to the sink over the network. The \emph{average coding delay} of a code over a network with respect to a class of traffics is the average of the coding delays of the code with respect to all the traffics.\footnote{The coding delay of a class of codes over a class of traffics is a random variable due to the randomness in both the code and the traffic. The average coding delay is the coding delay averaged out over the traffics but not the codes, and hence is a random variable due to the randomness in the code.}

Pakzad \emph{et al.} \cite{PFS:2005} studied the average coding delay of dense codes over the networks with deterministic regular transmission opportunities and Bernoulli losses, where the special case of two identical links in tandem was considered. The analysis however did not provide any insight about how the coding delay (which is random with respect to both the codes and the traffics) can deviate from the average coding delay (which is random with respect to the codes but not the traffics).

More recently, Dikaliotis \emph{et al.} \cite{DDHE:2009} studied both the average coding delay and the coding delay over networks similar to those in \cite{PFS:2005}, under the assumption that all the packets are innovative.\footnote{A collection of packets is ``innovative'' if their global encoding vectors are linearly independent.} This is not however a valid assumption in practice, where the field size is finite and can be as small as two.

In this paper, we study the coding delay and the average coding delay of dense codes over the field of size two ($\mathbb{F}_2$). The analysis however can be generalized to the finite fields of larger size. We consider both lossless networks and networks with Bernoulli losses. We also study both deterministic regular and Poisson transmission schedules.

The main contributions of this paper are:

\begin{itemize}
\item For networks with deterministic regular transmission opportunities and Bernoulli losses, we derive upper bounds on the average coding delay of dense codes tighter than what were presented in \cite{PFS:2005,DDHE:2009}.
\item We show that, for such networks, the coding delay may have a large deviation from the average coding delay in both cases of identical and non-identical links. For non-identical links, our upper bound on such a deviation is smaller than what was previously shown in \cite{DDHE:2009}. It is worth noting that, for identical links, upper bounding such a deviation has been an open problem (see \cite{DDHE:2009}).
\item We generalize the results to the networks with Poisson transmission schedules for both lossless networks and networks with Bernoulli losses.
\end{itemize}

\section{Network Model and Problem Setup}
We consider a line network of length $L$, where the $L+1$ nodes $\{v_i\}_{0\leq i\leq L}$ are connected in tandem. The underlying problem is unicast: The source node $v_0$ is given a message of $k$ vectors from a vector space over $\mathbb{F}_2$, and the sink node $v_{L}$ demands to have all the message vectors.

Each node transmits a (coded) packet at each transmission opportunity in discrete-time where the number of transmissions per transmission opportunity is one. The points in time at which the transmissions occur over each link follow a stochastic point process. The processes specifying the transmissions over different links are considered to be independent.

Each packet transmission is either successful or fails. In the latter case, the packet is erased. We consider two scenarios: (i) lossless, where all packet transmissions are successful, and (ii) lossy, where all packet transmissions are subject to independent erasures over the same link or different links. The traffic over a link is fully described by the processes describing the schedule of transmissions and by the loss model.

The links are assumed to be delay-free, i.e., the arrival time of a successful packet at a receiving node is the same as the departure time of the packet from the transmitting node.

In this paper, we use the notions of the coding delay and the average coding delay in a probabilistic fashion as follows:

\noindent For some fixed $0<\epsilon<1$, the \emph{$\epsilon$-constrained coding delay} of a class of codes over a network with a class of traffics is defined as the infimum value of $N \in \mathbb{Z}$ such that the coding delay of a randomly chosen code over the network with a randomly chosen traffic is larger than $N$ with probability (w.p.) bounded above by (b.a.b.) $\epsilon$. The \emph{$\epsilon$-constrained average coding delay} of a class of codes over a network with respect to a class of traffics is defined as the infimum value of $N \in \mathbb{Z}$ such that the average coding delay of a randomly chosen code over the network with respect to the class of traffics is larger than $N$ w.p. b.a.b. $\epsilon$. We often drop the term ``$\epsilon$-constrained'' for simplicity unless there is a danger of confusion.

The goal in this paper is to upper bound the coding delay and the average coding delay of dense codes over networks with two types of transmission schedules and two types of loss models specified below.

The transmission schedules are described by (i) a deterministic process where at each time unit there is a transmission opportunity at each node (such a schedule is referred to as \emph{deterministic regular}), or (ii) a Poisson process with parameter $\lambda_i: 0<\lambda_i<1$, over the $i\textsuperscript{th}$ link, where $\lambda_i$ is the average number of transmission opportunities per time unit.

The loss models are described by (i) a deterministic process where each packet transmission is successful (such a model is referred to as \emph{lossless}), or (ii) a Bernoulli process with parameter $p_i: 0<p_i<1$, over the $i\textsuperscript{th}$ link, where $p_i$ is the average number of successes per transmission opportunity.

\section{Deterministic Regular Lossless Traffic}\label{sec:LosslessRegularTraffic}
In a dense coding scheme, the source node, at each transmission opportunity, transmits a packet by randomly linearly combining the message vectors, and each non source non-sink (interior) node transmits a packet by randomly linearly combining its previously received packets. The vector of coefficients of the linear combination associated with a packet is called the \emph{local encoding vector} of the packet, and the vector of the coefficients representing the mapping between the message vectors and a coded packet is called the \emph{global encoding vector} of the packet. The global encoding vector of each packet is assumed to be included in the packet header. The sink node can recover all the message vectors as long as it receives an innovative collection of packets of the size equal to the number of message vectors at the source node.

The entries of the global encoding vectors of a collection of packets are independent and uniformly distributed (i.u.d.) Bernoulli random variables as long as the local encoding vectors of the packets are linearly independent. Such packets, called \emph{dense}, are of main importance in our analysis.

The first step is to lower bound the size of a maximal collection of dense packets at the sink node until a certain decoding time. We, next, lower bound the probability that the underlying collection includes a sufficient number of packets with linearly independent global encoding vectors.

Let $Q$ be a matrix over $\mathbb{F}_2$. A maximal collection of rows in $Q$ with i.u.d. entries is called \emph{dense}. The matrix $Q$ is called a \emph{dense matrix} if all its rows form a dense collection. We refer to the number of rows in a dense collection of rows in $Q$ as the \emph{density} of $Q$, denoted by $\mathcal{D}(Q)$, and refer to each row in such a collection as a \emph{dense row}.

Let $\mathcal{O}_{i}$ ($\mathcal{I}_i$) be the set of labels of the packets transmitted (received) by the $i\textsuperscript{th}$ node and let $\mathcal{D}_i$ be the set of labels of the dense packets at the $i\textsuperscript{th}$ node. Let $r$ and $d$ be the size of $\mathcal{O}_{i}$ and $\mathcal{D}_i$, respectively. The global encoding vectors of the received packets at a node form the rows of the decoding matrix at that node. Let $Q_{i+1}$ and $Q_{i}$ be the decoding matrices at the $(i+1)\textsuperscript{th}$ and $i\textsuperscript{th}$ nodes, respectively, and $T_{i}$ be a matrix over $\mathbb{F}_2$ such that $Q_{i+1}=T_{i}Q_{i}$. The rows of $T_i$ are the local encoding vectors of the packets transmitted by the $i\textsuperscript{th}$ node, i.e., $(T_i)_{n,j}=\lambda_{n,j}$, $\forall {n}\in \mathcal{O}_{i}$ and $\forall {j}\in \mathcal{I}_{i}$, where $\lambda_n$ is the local encoding vector of the $n\textsuperscript{th}$ packet. Let $Q'_{i}$ be $Q_{i}$ restricted to its dense rows, i.e., $Q'_{i}$ is dense and has $d$ rows ($\mathcal{D}(Q_{i})=d$). We can write $Q_{i+1}=T'_{i}Q'_{i}$, where $T'_{i}$, the \emph{transfer matrix} at the $i\textsuperscript{th}$ node, is a matrix over $\mathbb{F}_2$ with $d$ columns: $(T'_{i})_{{n},j}={\lambda}_{{n},j}+\sum_{\ell\in \mathcal{I}_{i}\setminus \mathcal{D}_{i}}\lambda_{n,\ell}\gamma_{\ell,j}$, $\forall n\in \mathcal{O}_i$, $\forall j\in\mathcal{D}_i$ and $\{\gamma_{\ell,j}\}$ are in $\mathbb{F}_2$ satisfying $\sum_{j\in \mathcal{D}_i}\gamma_{\ell,j}\lambda_{j,k}=\lambda_{\ell,k}$, $\forall k\in\mathcal{I}_{i}$.


The $n\textsuperscript{th}$ row of $T^{\prime}_{i}$ indicates the labels of dense packets at the $i\textsuperscript{th}$ node which contribute to the $n\textsuperscript{th}$ packet sent by the $i\textsuperscript{th}$ node, and the $j\textsuperscript{th}$ column of $T^{\prime}_{i}$ indicates the labels of packets sent by the $i\textsuperscript{th}$ node to which the $j\textsuperscript{th}$ dense packet contributes. Let $\mathcal{T}_{i_{\text{row}}}^{\prime (n)}$ ($\mathcal{T}_{i_{\text{col}}}^{\prime (j)}$) be the set of labels of i.u.d. entries in the $n\textsuperscript{th}$ row ($j\textsuperscript{th}$ column) of $T^{\prime}_i$. Thus, $|{\mathcal{T}}_{i_{\text{row}}}^{\prime (n)}|\geq \max\{n-r+d,0\}$ (in particular, the first $\max\{n-r+d,0\}$ entries of the $n\textsuperscript{th}$ row are i.u.d.). Similarly, $|{\mathcal{T}}_{i_{\text{col}}}^{\prime (j)}|\geq d-j+1$ (in particular, the last $d-j+1$ entries of the $j\textsuperscript{th}$ column are i.u.d.).

Let $\text{rank}(T)$ denote the rank of a matrix $T$ over $\mathbb{F}_2$. The following result is then useful to lower bound the density of the decoding matrix $Q_{i+1}$ in terms of $\text{rank}(T^{\prime}_{i})$.\footnote{The proofs of the lemmas in this section can be found in~\cite{HBJ:2011}.}


\begin{lemma}\label{lem:DensityTM} Let $M$ be a dense matrix over $\mathbb{F}_2$, and $T$ be a matrix over $\mathbb{F}_2$, where the number of rows in $M$ and the number of columns in $T$ are equal. If $\text{rank}(T)\geq \gamma$, then $\mathcal{D}(TM)\geq \gamma$.\end{lemma}

The rank of a matrix $T$ similar to that of the transfer matrix $T'$ specified earlier can be lower bounded as follows.

\begin{lemma}\label{lem:SingleT} Let $T$ be an $n\times d$ ($d\leq n$) matrix over $\mathbb{F}_2$ such that for any $1\leq j\leq d$, at least $d-j+1$ entries of its $j\textsuperscript{th}$ column are i.u.d.. For every integer $0\leq \gamma\leq d-1$, \[\Pr\{\text{rank}(T)<d-\gamma\}\leq (d-\gamma)2^{-(\gamma+1)}.\]\end{lemma}

\begin{proof} For any integer $0\leq \gamma\leq d-1$, let $T'$ be $T$ restricted to its first $d-\gamma$ columns. Since $T'$ is an $n\times (d-\gamma)$ sub-matrix of $T$, $\Pr\{r(T)<d-\gamma\}\leq \Pr\{r(T')<d-\gamma\}$. Suppose that $r(T')<d-\gamma$. Then there exists a nonzero column vector $v$ of length $d-\gamma$ over $\mathbb{F}_2$ such that the column vector $T'v$ of length $n$ is an all-zero vector. For an integer $1\leq j\leq d-\gamma$, suppose that the first non-zero entry of $v$ is the $j\textsuperscript{th}$. There exist $2^{d-\gamma-j}$ such vectors. Since there exist at least $d-j+1$ i.u.d. entries in the $j\textsuperscript{th}$ column of $T'$, there exist at least $d-j+1$ i.u.d. entries in the vector $T'v$. The probability that all these entries are zero is $2^{-d+j-1}$, and thus the probability that $T'v$ is an all-zero vector given that the first nonzero entry of $v$ is the $j\textsuperscript{th}$ is b.a.b. $2^{-d+j-1}$. Taking a union bound over all such vectors $v$, the probability that $T'v$ is an all-zero vector is $2^{d-\gamma-j}\times 2^{-d+j-1}=2^{-\gamma-1}$. Taking a union bound over all $j$: $1\leq j\leq d-\gamma$, the probability that $T'v$ is an all-zero vector is b.a.b. $(d-\gamma)2^{-(\gamma+1)}$.\end{proof}

The preceding lemma is a special case of what we state in the following. The latter is useful in order to generalize the results on one transmission per opportunity to multiple transmissions per opportunity.

For given integers $w$ and $r$, let $T_{i,j}$ be an $r\times r$ dense matrix over $\mathbb{F}_2$, $\forall i,j: 1\leq j\leq i\leq w$, and $T_{i,j}$ be an all-zero $r\times r$ matrix, $\forall i,j: 1\leq i<j\leq w$. Let $T=[T_{i,j}]_{1\leq i,j\leq w}$, and $n\doteq wr$.

\begin{lemma}\label{lem:SquareT}Let $T$ be defined as above. For every integer $0\leq \gamma\leq n-1$, \[\Pr\{r(T)<n-\gamma\}\leq \left\lceil\frac{n-\gamma}{r}\right\rceil \left(1-2^{-r}\right) 2^{-\gamma}.\]\end{lemma}

\begin{proof} Let $T'$ and $v$ be defined as in the proof of Lemma~\ref{lem:SingleT}. Fix an integer $1\leq j\leq n-\gamma$. Suppose that the first non-zero entry of $v$ is the $j\textsuperscript{th}$. There exist $2^{n-\gamma-j}$ such vectors. Let $\tau$ be the largest integer smaller than $j/r$. The $j\textsuperscript{th}$ column has at least $n-\tau r$ i.u.d. entries, and hence there exist at least $n-\tau r$ i.u.d. entries in the vector $T'v$. These entries are all zero w.p. $2^{\tau r-n}$, and $T'v$ is all-zero given such $v$ w.p. b.a.b. $2^{\tau r-n}$. Taking a union bound over all such vectors, the latter probability is $2^{\tau r-\gamma-j}$. Taking a union bound over $j$, $T'v$ is all-zero w.p. b.a.b. $2^{-\gamma}\sum_{1\leq j\leq n-\gamma} 2^{\tau r-j}$, noting that $\tau$ depends on $j$. We shall upper bound the preceding sum by rewriting it as: $\sum_{0<j\leq r}2^{-j}+\sum_{r<j\leq 2r}2^{r-j}$ $+\cdots+$ $\sum_{(u-1) r<j\leq n-\gamma}2^{(u-1)r-j} = \sum_{0<j\leq r}2^{-j}+\sum_{0<j\leq r}2^{-j}+\cdots+\sum_{0<j\leq n-\gamma-(u-1) r}2^{-j} \leq u \sum_{0<j\leq r}2^{-j} = u \left(1-2^{-r}\right)$, where $u = \left\lceil {(n-\gamma)}/{r}\right\rceil$. This completes the proof.\end{proof}

Let $(0,N_T]$ be the period of time over which the transmissions occur. The decoding matrix at the first internal node ($v_1$) is dense and its density is equal to the number of packets at the node until time $N_T$, i.e., $\mathcal{D}(Q_1)=N_T$. The density of the decoding matrix at the other non-source nodes is bounded from below as follows by applying the preceding lemmas.

\begin{lemma}\label{lem:DensitySinkRelative}For every $1< i\leq L$, the inequality \[\mathcal{D}(Q_{i})\geq \mathcal{D}(Q_{i-1})-\log \mathcal{D}(Q_{i-1})-\log(1/\epsilon)\] fails w.p. b.a.b. $\epsilon$.\end{lemma}

By combining the result of Lemma~\ref{lem:DensitySinkRelative} with $\mathcal{D}(Q_1) = N_T$, we can derive the following result.

\begin{lemma}\label{lem:DensitySink} Suppose that a dense code is applied over a line network of $L$ links with deterministic regular lossless traffics until time $N_T$. Then, the inequality \[\mathcal{D}(Q_{L})\geq N_T-L\log(N_T L/\epsilon)\] fails w.p. b.a.b. $\epsilon$.\end{lemma}

Now, we lower bound the probability that the collection of dense packets at the sink node includes an innovative sub-collection of size $k$. This itself lower bounds the probability that a dense code succeeds.

\begin{lemma}\label{lem:DenseRankProb}Let $M$ be an $n\times k$ ($k\leq n$) dense matrix over $\mathbb{F}_2$. For every $0<\epsilon<1$, \[\Pr\{\text{rank}(M)<k\}\leq \epsilon,\] if $k\leq n-\log(1/\epsilon)$.\end{lemma}

The following result upper bounds the coding delay by putting together the results of Lemmas~\ref{lem:DensitySink} and~\ref{lem:DenseRankProb}.

\begin{theorem}\label{thm:DenseCodeRegularLossless}The $\epsilon$-constrained coding delay of a dense code over a line network of $L$ links with deterministic regular lossless traffics is b.a.b. \[k+L\log(L/\epsilon)+\log(1/\epsilon)+L+1.\]\end{theorem}

\section{Deterministic Regular Traffic with Bernoulli Losses}\label{sec:BernoulliLossRegularTraffic}
\subsection{Identical Links}
In this case, the Bernoulli parameters $\{p_i\}_{1\leq i\leq L}$ are all the same, and equal to $p$. Similar to the analysis of the previous case, in the case of the deterministic regular traffic with Bernoulli losses, we need to track the number of dense packets through the network.

The density of the decoding matrix at the receiving node of a link depends on the density of the decoding matrix and the rank of the transfer matrix at the transmitting node of the link. The rank of a matrix is a function of its structure, and the structure of the transfer matrix at a node depends on the number of dense packet arrivals at the node and the number of packet departures from the node before or after any given time. Such parameters depend on the transmission schedule and the loss model of the link, and are therefore random variables. It is however not straightforward to find the distribution of such random variables. We rather adopt a probabilistic technique to lower bound the rank of the transfer matrices as follows.

We split the time interval $(0,N_T]$ into a number of disjoint subintervals (partitions) of the same length. The arrivals in the first $j$ partitions occur before the departures in the $(j+1)\textsuperscript{th}$ partition. Thus the number of arrivals before a given point in time within the $(j+1)\textsuperscript{th}$ partition is bounded from below by the sum of the number of arrivals in the first $j$ partitions. Such a method of counting is however suboptimal since there might be some extra arrivals in the $(j+1)\textsuperscript{th}$ partition before some points in time within the same partition. To control the impact of suboptimality, the length of the partitions thus needs to be chosen with some care.\footnote{\label{ftnt:PartitionLength} On one hand, the length of the partitions needs to be sufficiently small such that there is not a large number of arrivals in one partition with respect to the total number of arrivals in all the partitions. This should be the case because ignoring a subset of arrivals in one partition should not cause a significant difference in the number of arrivals before each point in time within the same partition. On the other hand, the partitions need to be long enough such that the deviation of the number of arrivals from the expectation in one partition is negligible in comparison with the expectation itself.}



Let $w$ be the number of partitions of the interval $(0,N_T]$. Let $I_{ij}$ be the $j\textsuperscript{th}$ partition pertaining to the $i\textsuperscript{th}$ link for all $i$ and $j$. We start off with lower bounding the number of packets in $I_{ij}$. Let $\varphi_{ij}$ be the number of packets in $I_{ij}$. The length of the partition $I_{ij}$ is $N_{T}/w$. Thus, $\varphi_{ij}$ is a binomial random variable with the expected value $\varphi \doteq pN_T /w$.




Hereafter, for the ease of exposition, let us denote $x/2$ by $\dot{x}$, for every $x\in\mathbb{R}$. By applying the Chernoff bound, one can show that the inequality \[\varphi_{ij}\geq r\doteq \left(1-\gamma^{*}\right)\varphi\] fails w.p. b.a.b. $\dot{\epsilon}$, so long as $\gamma^{*}$ is chosen such that $r$ is an integer, and $\gamma^{*}$ goes to $0$ as $N_T$ goes to infinity, where \begin{equation}\label{eq:gammastar}\gamma^{*}\sim \left(\frac{2}{\varphi}\ln\frac{2}{\epsilon}\right)^{\frac{1}{2}}.\end{equation}





We focus on the set of all packets over the $i\textsuperscript{th}$ link in the \emph{active} partitions: $I_{ij}$ is `active' if $i\leq j\leq w-L+i$. Such a partition is active in the sense that (i) there exists some other partition over the upper link so that all its packets arrive before the departure of all the packets in the underlying active partition, and (ii) there exists some other partition over the lower link so that all its packets depart after the arrival of all the packets in the underlying active partition.

Let $w_T$ denote the total number of active partitions. It is easy to see that $w_T = L(w-L+1)$. We select $r$ packets in each active partition and ignore the rest. This method of selection fails if the number of packets in some active partition is less than $r$. Clearly, the failure occurs w.p. b.a.b. $w_T \dot{\epsilon}$.

We shall lower bound the number of dense packets in active partitions. Before explaining the lower bounding technique, let us first state two lemmas which will be useful to lower bound the rank of the transfer matrix at each node (depending on whether the number of dense packet arrivals at the node in a partition is larger or smaller than the number of packet departures from the node in the same partition).



For given integers $w$, $r$ and $\{r_j\}_{1\leq j\leq w}$ ($0\leq r_j\leq r$), let $T_{i,j}$ be defined as follows: $T_{i,j}$ is an $r\times r_j$ dense matrix over $\mathbb{F}_2$, if $1\leq j\leq i\leq w$; or an arbitrary $r\times r_j$ matrix over $\mathbb{F}_2$, otherwise (i.e., if $1\leq i<j\leq w$). Let $T=[T_{i,j}]_{1\leq i,j\leq w}$, and $n\doteq \sum_{1\leq j\leq w}r_j$.

\begin{lemma}\label{lem:VerticalT}Let $T$ be defined as above. For every integer $0\leq \gamma\leq n-1$, \[\Pr\{r(T)<n-\gamma\}\leq u \left(1-2^{-r_{\text{max}}}\right) 2^{-\gamma+n-wr+(r-r_{\text{min}})(u-1)},\] where $r_{\text{max}}=\max_j r_j$, $r_{\text{min}}=\min_j r_j$, and $u=\left\lceil{(n-\gamma)}/{r_{\text{min}}}\right\rceil$.\end{lemma}

\begin{proof} Let $T'$ and $v$ be defined as in the proof of Lemma~\ref{lem:SingleT}. Let us define $r_0\doteq 0$ for convenience. For a given integer $1\leq j\leq n-\gamma$, define $\tau: \sum_{0\leq i\leq \tau}r_i<j\leq \sum_{0\leq i\leq \tau+1}r_i$. Further, define $\tau_{\text{max}}: \sum_{0\leq i\leq \tau_{\text{max}}}r_i<n-\gamma\leq\sum_{0\leq i\leq \tau_{\text{max}}+1}r_i$ ($0\leq \tau_{\text{max}}<w$). By definition, $\tau_{\text{max}}\leq\min\{w,u-1\}$. For every $0\leq \tau\leq \tau_{\text{max}}$, define $s_{\tau}=\sum_{0\leq i\leq \tau}r_i$. The $j\textsuperscript{th}$ column of $T'$ has at least $(w-\tau)r$ i.u.d. entries, and hence the vector $T'v$ has at least $(w-\tau)r$ i.u.d. entries. Thus, $T'v$ is all-zero w.p. b.a.b. $2^{-\gamma+n-wr}\sum_{1\leq j\leq n-\gamma}2^{\tau r-j}$, noting that $\tau$ depends on $j$. We rewrite the sum as:
\begin{eqnarray*}
&& \hspace*{-1 cm} \sum_{0<j\leq s_1} 2^{-j}+\sum_{s_1<j\leq s_2}2^{r-j}+ \\
& & \cdots+\sum_{s_{\tau_{\text{max}}}<j\leq n-\gamma}2^{\tau_{\text{max}}r-j} = \\
& & \sum_{0<j\leq r_1}2^{-j}+2^{r-s_1}\sum_{0<j\leq r_2}2^{-j}+ \\
& & \cdots+2^{\tau_{\text{max}}r-s_{\tau_{\text{max}}}}\sum_{0<j\leq n-\gamma-s_{\tau_{\text{max}}}}2^{-j}\leq \\
& & \sum_{0<j\leq r_{\text{max}}}2^{-j} + 2^{r-s_1}\sum_{0<j\leq r_{\text{max}}}2^{-j}+ \\
& & \cdots+2^{\tau_{\text{max}}r-s_{\tau_{\text{max}}}}\sum_{0<j\leq r_{\text{max}}}2^{-j}= \\
& & \sum_{0<j\leq r_{\text{max}}}2^{-j}\sum_{0\leq \tau'\leq \tau_{\text{max}}}2^{\tau' r-s_{\tau'}}\leq \\
& & \sum_{0<j\leq r_{\text{max}}}2^{-j}\sum_{0\leq \tau'\leq \tau_{\text{max}}}2^{(r-r_{\text{min}})\tau'}= \\
& & (1-2^{-r_{\text{max}}})\sum_{0\leq \tau'\leq \tau_{\text{max}}}2^{(r-r_{\text{min}})\tau'}.
\end{eqnarray*} The series $\sum_{0\leq \tau'\leq \tau_{\text{max}}}2^{(r-r_{\text{min}})\tau'}$ converges from below to $(\tau_{\text{max}}+1)2^{(r-r_{\text{min}})\tau_{\text{max}}}$ if $r-r_{\text{min}}$ goes to infinity. Thus the following is always true: $(1-2^{-r_{\text{max}}})$ $\sum_{0\leq \tau'\leq \tau_{\text{max}}}2^{(r-r_{\text{min}})\tau'}$ $\leq$ $(\tau_{\text{max}}+1)(1-2^{-r_{\text{max}}})2^{(r-r_{\text{min}})\tau_{\text{max}}}$ $\leq u(1-2^{-r_{\text{max}}})2^{(r-r_{\text{min}})(u-1)}$. This proves the lemma.\end{proof}

For given integers $w$, $r$ and $\{r_j\}_{1\leq j\leq w}$ ($r\leq r_j$), let $T_{i,j}$ be defined as follows: $T_{i,j}$ is an $r\times r_j$ dense matrix over $\mathbb{F}_2$, if $1\leq j\leq i\leq w$; or an arbitrary $r\times r_j$ matrix over $\mathbb{F}_2$, otherwise (i.e., if $1\leq i<j\leq w$). Let $T=[T_{i,j}]_{1\leq i,j\leq w}$, and $n\doteq w r$.

\begin{lemma}\label{lem:HorizontalT} Let $T$ be defined as above. For every integer $0\leq \gamma\leq n-1$, \[\Pr\{r(T)<n-\gamma\}\leq u \left(1-2^{-r}\right) 2^{-\gamma+n-wr_{\text{min}}+(r_{\text{min}}-r)(u-1)},\] where $u=\left\lceil{(n-\gamma)}/{r}\right\rceil$.\end{lemma}

\begin{proof} We start the proof by noting that $T$ has a smaller number of rows than columns, and the minimum number of rows and columns gives an upper bound on the rank of the matrix. Let $T'$ be $T$ restricted to its last $n-\gamma$ rows. For every $0\leq\tau\leq w$, define $s_{\tau}=\sum_{0\leq j\leq w-\tau}r_j$. Thus, $T'$ is of size ${(n-\gamma)\times s_0}$. Suppose that there exists a nonzero row vector $v$ of length $n-\gamma$ whose entries are over $\mathbb{F}_2$, and its first nonzero entry is the $j\textsuperscript{th}$, and the row vector $vT'$ is all-zero. There are $2^{n-\gamma-j}$ such vectors. Let $\tau$ be the largest integer smaller than $j/r$. The $j\textsuperscript{th}$ row of $T'$ has at least $s_{\tau}$ i.u.d. entries, and hence the vector $vT'$ has at least $s_{\tau}$ i.u.d. entries. Thus, $vT'$ is all-zero w.p. b.a.b. $2^{-\gamma+n}\sum_{1\leq j\leq n-\gamma}2^{-j-s_{\tau}}$. By definition, $s_{\tau}\geq (w-\tau)r_{\text{min}}$, and the preceding sum can thus be upper bounded as follows: $\sum_{1\leq j\leq n-\gamma}2^{-j-s_{\tau}}\leq$ $\sum_{1\leq j\leq n-\gamma}2^{-j-(w-\tau)r_{\text{min}}}$. The latter sum can be rewritten itself as:
\begin{eqnarray*}&& \hspace*{-1 cm} \sum_{0<j\leq r}2^{-j-wr_{\text{min}}}+\sum_{r<j\leq 2r}2^{-j-(w-1)r_{\text{min}}}+\\
& & \cdots +\sum_{(u-1)r<j\leq n-\gamma}2^{-j-(w-u+1)r_{\text{min}}} =  \\
& & 2^{-wr_{\text{min}}}\sum_{0<j\leq r}2^{-j}+2^{-(w-1)r_{\text{min}}-r}\sum_{0<j\leq r}2^{-j}+ \\
& & \cdots+2^{-(w-1)r_{\text{min}}-(u-1)r}\sum_{0<j\leq n-\gamma-(u-1)r}2^{-j} \leq \\
& & 2^{-wr_{\text{min}}}\sum_{0<j\leq m}2^{-j}\sum_{0\leq \tau'\leq u-1}2^{(r_{\text{min}}-r)\tau'}= \\
& & (1-2^{-r})2^{-wr_{\text{min}}}\sum_{0\leq \tau'\leq u-1}2^{(r_{\text{min}}-r)\tau'}.\end{eqnarray*} The last sum is bounded from above by $u \cdot 2^{(r_{\text{min}}-r)(u-1)}$, and this completes the proof. \end{proof}

For every $1<i\leq L$, and $1\leq j\leq w-L+1$, the number of dense packets in the first $j$ active partitions over the $i\textsuperscript{th}$ link can be lower bounded as follows: For every $1\leq l\leq j$, suppose that the number of dense packets in the first $l$ active partitions over the $(i-1)\textsuperscript{th}$ link is already lower bounded. Let $T$ be the transfer matrix at the $i\textsuperscript{th}$ node, restricted to the successful packet transmissions within the first $j$ active partitions over the $i\textsuperscript{th}$ link (the number of such packets in each partition is already lower bounded). Then, it can be shown that $T$ includes a sub-matrix $T'$ with a structure similar to that in Lemma~\ref{lem:VerticalT} or the one in Lemma~\ref{lem:HorizontalT}.\footnote{In the case of identical links, the transfer matrix at each node includes a sub-matrix similar to that in Lemma~\ref{lem:VerticalT}. However, in the case of non-identical links, depending on the traffic parameters, the transfer matrix at a node might include a sub-matrix similar to that in Lemma~\ref{lem:VerticalT} or the one in Lemma~\ref{lem:HorizontalT}.} By applying the proper lemma, the rank of the transfer matrix at the $i\textsuperscript{th}$ node, and consequently, by applying Lemma~\ref{lem:DensityTM}, the number of dense packets in the first $j$ active partitions over the $i\textsuperscript{th}$ link can be lower bounded.

Note that, because of its recursive nature, the above algorithm lower bounds the number of dense packets in the first $j$ active partitions over the $i\textsuperscript{th}$ link as a function of the number of dense packets in the active partitions pertaining to the first link. Further, the packets over the first link are all dense (by the definition of the dense packets), and hence by using the recursion, the following results can be derived.

Let $\mathcal{D}(Q_i^j)$ be the number of dense packets in the first $j$ active partitions over the $i\textsuperscript{th}$ link. Let $\mathcal{D}_{p}(Q_i^j)$ lower bound $\mathcal{D}(Q_i^j)$ such that $\mathcal{D}(Q_i^j)<\mathcal{D}_{p}(Q_i^j)$ w.p. b.a.b. $\dot{\epsilon}$, given $\mathcal{D}(Q_{s}^{\tau})\geq\mathcal{D}_{p}(Q_{s}^{\tau})$, for every $1\leq s\leq i$ and $1\leq \tau\leq j$, except $(s,\tau)=(i,j)$.\footnote{$\mathcal{D}_{p}(Q_i^j)$ is a ``proper'' lower bound on $\mathcal{D}(Q_i^j)$ for the purpose of the analysis in this paper and hence the subscript ``$p$.''} We define $r_{ij}$ in a recursive fashion as the largest non-negative integer satisfying $r_{ij}\leq \mathcal{D}_{p}(Q_i^j)-\sum_{1\leq \tau<j}r_{i\tau}$.

We construct a collection of dense packets at the $i\textsuperscript{th}$ node as follows: starting with an empty collection (at the step zero), for every $1\leq j\leq w-L+1$, at the $j\textsuperscript{th}$ step, we expose the packets in the active partitions over the $i\textsuperscript{th}$ link in order, one by one. We add a packet to the collection whenever the packet is dense (with respect to the current collection), until revealing $r_{ij}$ new dense packets. The size of such a collection lower bounds the number of dense codes at the $i\textsuperscript{th}$ node, and in order to study the structure of the transfer matrix at this node, we consider the packets in the subsets of the underlying collection, each subset pertaining to one of the collection steps, and ignore the rest of packets.

Clearly, $\mathcal{D}(Q_1^j)\geq rj$, $\forall j: 1\leq j\leq w-L+1$ (since $r$ packets are selected in each partition). For any other values of $i$ and $j$, $\mathcal{D}(Q_i^j)$ is lower bounded as follows.

\begin{lemma}\label{lem:Omegaii}For every $1< i\leq L$, \[\mathcal{D}(Q_i^1)\geq r-\log(1/\epsilon)-\log i - 1\] fails w.p. b.a.b. $\dot{\epsilon}$.\end{lemma}

\begin{proof} Fix $1<i\leq L$. Let $T$ be the transfer matrix at the starting node of the $i\textsuperscript{th}$ link. Let $T'$ be $T$ restricted to the packets in the first active partition over the $i\textsuperscript{th}$ link. For every $1<s<i$, suppose $\mathcal{D}(Q_s^1)\geq \mathcal{D}_p(Q_s^1)$, where $\mathcal{D}_p(Q_s^1)=r-\log(1/\epsilon)-1$, and $\mathcal{D}(Q_1^1)=\mathcal{D}_p(Q_1^1)=r$. Then, $T'$ includes an $r\times r_1$ dense sub-matrix.\footnote{We often drop the subscript $i$ in the notation $r_{ij}$ unless there is a danger of confusion.} Thus by applying Lemma~\ref{lem:VerticalT}, for every $0\leq \gamma\leq r_1-1$, $\Pr\{r(T')<r_1-\gamma\}\leq u (1-2^{-r_1}) 2^{-\gamma+r_1-r+(r-r_1)(u-1)}=(1-2^{-r_1})2^{-\gamma+r_1-r}$, since $u=\lceil(r_1-\gamma)/r_1\rceil=1$. Taking $\gamma=\log(1/\epsilon)+r_1-r+1$, it follows that $\Pr\{r(T')<r_1-\gamma\}\leq \dot{\epsilon}$. By applying Lemma~\ref{lem:DensityTM}, $\mathcal{D}(Q_i^1)<r-\log(1/\epsilon)-1$ w.p. b.a.b. $\dot{\epsilon}$. Thus, $\mathcal{D}_p(Q_i^1)=r-\log(1/\epsilon)-1$. Taking a union bound over the first $i$ links, $\mathcal{D}(Q_i^1)<r-\log(1/\epsilon)-\log i-1$ w.p. b.a.b. $\dot{\epsilon}$. \end{proof}

\begin{lemma}\label{lem:Omegaij} For every $1<i\leq L$, and $1<j\leq w-L+1$, \[\mathcal{D}(Q_i^j)\geq rj-\mathcal{L}_{ij}\] fails w.p. b.a.b. $\dot{\epsilon}$, so long as $\log({w_T}/{\epsilon})=o(r)$, where $\mathcal{L}_{ij} = j(1+o(1))(\log(ij/\epsilon)+1)+\log((j(1+o(1))+1)/\epsilon)+\log(ij)+1$, and the $o(1)$ term is $(\log(ij/\epsilon)+1)/r$.\end{lemma}

\begin{proof} Fix $1<i\leq L$. For every $1<s\leq i$ and $1<\tau\leq j$, except $(s,\tau)=(i,j)$, suppose $\mathcal{D}(Q_s^{\tau})\geq \mathcal{D}_p(Q_s^{\tau})$, where $\mathcal{D}_p(Q_s^{\tau})=r\tau-\tau(1+o(1))(\log(1/\epsilon)+1)-\log((\tau(1+o(1))+1)/\epsilon)-1$, and the $o(1)$ term is $(\log(1/\epsilon)+1)/r$, and $\mathcal{D}(Q_s^{1})\geq \mathcal{D}_p(Q_s^{1})$, where $\mathcal{D}_p(Q_s^{1})=r-\log(1/\epsilon)-1$. Let $r_{\text{min}}=\min_{\tau}r_{\tau}$, and $r_{\text{max}}=\max_{\tau}r_{\tau}$. Let $r_{\tau}=r_{i-1,\tau}$, for every $1\leq \tau\leq j$, and $n=\mathcal{D}_p(Q_{i-1}^{j})=\sum_{1\leq\tau\leq j}r_{\tau}$. Let us define $T$ as in the proof of Lemma~\ref{lem:Omegaii}. Let $T'$ be $T$ restricted to the packets in the first $j$ active partitions over the $i\textsuperscript{th}$ link. Then, $T'$ includes an $rj\times n$ sub-matrix with a structure similar to the matrix $T$ as in Lemma~\ref{lem:VerticalT}. Thus by applying Lemma~\ref{lem:VerticalT}, for every $0\leq \gamma\leq n-1$, $\Pr\{r(T')<n-\gamma\}\leq u(1-2^{-r_{\text{max}}})2^{-\gamma+n-rj+(r-r_{\text{min}})(u-1)}$, where $u=\lceil (n-\gamma)/r_{\text{min}}\rceil$. It is not difficult to see that, by our method of constructing the dense collection, it follows that $r_{\text{min}}=r_1$. Further by applying Lemma~\ref{lem:Omegaii}, $r_1=\mathcal{D}_p(Q_{i-1}^1)=r-\log(1/\epsilon)-1$. Thus, $u\leq \lceil rj/r_{1}\rceil=\lceil (1+o(1))j\rceil\leq (1+o(1))j+1$, since $r_1=r(1-o(1))$, given $\log(w_T/\epsilon)=o(r)$, where the $o(1)$ term is $(\log(1/\epsilon)+1)/r$. Taking $\gamma=n-rj+(1+o(1))j(\log(1/\epsilon)+1)+\log(((1+o(1))j+1)/\epsilon)+1$, it follows that $\Pr\{r(T')<n-\gamma\}\leq \dot{\epsilon}$. Now, by applying Lemma~\ref{lem:DensityTM}, $\mathcal{D}(Q_i^j)<n-\gamma$ w.p. b.a.b. $\dot{\epsilon}$. Thus, $\mathcal{D}_p(Q_i^j)=n-\gamma$. Taking a union bound over the first $j$ active partitions of the first $i$ links, $\mathcal{D}(Q_i^j)<rj-(1+o(1))j(\log(ij/\epsilon)+1)-\log(((1+o(1))j+1)/\epsilon)-\log(ij)-1$ w.p. b.a.b. $\dot{\epsilon}$, where the $o(1)$ term is $(\log(ij/\epsilon)+1)/r$. This completes the proof.\end{proof}


The result of Lemma~\ref{lem:Omegaij} lower bounds the number of dense packets at the sink node as follows.

\begin{lemma}\label{lem:DensityBound}The inequality \begin{eqnarray}\label{eq:DensityBound}
 \lefteqn{\mathcal{D}(Q_L) \geq w_T\varphi/L -w_T\varphi/L \sqrt{(1/\dot{\varphi})\log(w_T/\dot{\epsilon})} -} \nonumber\\
   && {}-(w_T/L)\log(w_T/\dot{\epsilon})-(w_T/L\varphi)\log^{2}(w_T/\epsilon)- \nonumber \\
   && -(w_T/L\varphi)\log(w_T/\epsilon) -\log(w_T/{\epsilon}) - \nonumber\\
   && - \log(w_T/L) - 1
\end{eqnarray} fails w.p. b.a.b. $\epsilon$, where $w\sim\left(p N_T L^2/\log(N_T L/\epsilon)\right)^{{1}/{3}}$.\end{lemma}


\begin{proof} For the ease of exposition, let $v=w_T/L$. Lemma~\ref{lem:Omegaij} gives a lower bound on $\mathcal{D}(Q_L^{v})$. Thus, we can write: $\mathcal{D}(Q_L)\geq\mathcal{D}(Q_L^{v})\geq rv-v(1+o(1))\left(\log(w_T/\epsilon)+1\right)-\log((v (1+o(1)))/\epsilon)-\log w_T-1$, where the $o(1)$ term is $\left(\log(w_T/\epsilon)\right)/r$. This bound fails w.p. b.a.b. $\dot{\epsilon}$, given the success of the packet collection process. Further, $r=(1-o(1))\varphi$, where the $o(1)$ term is $\sqrt{(1/\dot{\varphi})\ln(w_T/\dot{\epsilon})}$. Thus, $\mathcal{D}(Q_L)\geq \varphi v -o(\varphi v)-v\log(w_T/\epsilon)-v-o(v\log(w_T/\epsilon))-o(v)-\log(v/\epsilon)-\log w_T-1$ fails w.p. b.a.b. ${\epsilon}$, where $o(\varphi v)\sim O(\varphi v \sqrt{(1/{\varphi})\log(w_T/{\epsilon})})$, and $o(v)\sim (v/\varphi)\log(w_T/\epsilon)$. By considering the dominant terms, the right-hand side of the last inequality can be written as \begin{equation}\label{eq:Temp1}\varphi v- O(v \sqrt{\varphi\log(w_T/\epsilon)})-O(v\log(w_T/\epsilon)).\end{equation} We now replace $\varphi$ and $v$ by $p N_T/w$ and $w$ ($v\sim w$), respectively, and rewrite the above as \begin{eqnarray}\label{eq:Temp2} && \hspace*{-1 cm} p N_T - O(p N_T L/w) - \nonumber \\
& & O(\sqrt{p N_T w\log(w L/\epsilon)})-O(w\log(w L/\epsilon)).\end{eqnarray} We would like $\mathcal{D}(Q_L)\geq(1-o(1))p N_T$. Each $O(.)$ term needs to be $o(p N_T)$. When considering the third term, it is easy to show that we need $w\log(w L/\epsilon)=o(p N_T)$. When this condition holds, the second term dominates the third one. We need to specify $w$ with some care in order to minimize \begin{equation}\label{eq:Temp3} O(p N_T L/w) + O(\sqrt{p N_T w\log(w L/\epsilon)}).\end{equation} We define $w$ as \[\sqrt[3]{\frac{p N_T L^2}{\log(N_T L/\epsilon)}}.\] This choice of $w$ ensures that the $O(.)$ terms are $o(p N_T)$.\end{proof}

Let $n_T$ be equal to the right-hand side of the inequality \eqref{eq:DensityBound}. Thus, $Q_L$ fails to include an $n_T\times k$ dense sub-matrix w.p. b.a.b. $\epsilon$. By applying Lemma~\ref{lem:DenseRankProb}, the probability of $\{\text{rank}(Q_L)<k\}$ is b.a.b. $\epsilon$, so long as $k\leq n_T-\log(1/\epsilon)$. We replace $\epsilon$ with $\dot{\epsilon}$ everywhere. Then, a dense code fails to transmit $k$ message vectors w.p. b.a.b. $\epsilon$, so long as $k\leq n_T - \log(1/\epsilon)-1$.

In the asymptotic setting as $N_T$ goes to infinity, $n_T$ can be written as \[p N_T-(1+o(1))(p N_T L/w+\sqrt{p N_T w\log(wL/\epsilon)}+w\log(wL/\epsilon)).\] We rewrite the last inequality as \[k\leq p N_T-(1+o(1))(p N_T L/w+\sqrt{p N_T w\log(wL/\epsilon)}+w\log(wL/\epsilon)) -\log(1/\epsilon)-1.\] Let $k_{\text{max}}$ be the largest integer $k$ satisfying this inequality. Thus, $k_{\text{max}}\sim p N_T$, as $n_T\sim p N_T$ and $\log(1/\dot{\epsilon})=o(n_T)$. The following result can be shown by replacing $N_T$ with $k/p$ in the right-hand side of the latter inequality.

\begin{theorem}\label{thm:DenseCodesPoissonActual} The $\epsilon$-constrained coding delay of a dense code over a line network of $L$ identical links with regular traffics and Bernoulli losses with parameter $p$ is b.a.b. \[\frac{1}{p}\left(k+(1+o(1))\left(\frac{kL}{w}+\sqrt{k\left(w\log\frac{wL}{\epsilon}\right)}+w\log\frac{wL}{\epsilon}\right)\right)\] where $w\sim \left(k L^2/\log(k L/p \epsilon)\right)^{1/3}$, and the $o(1)$ term goes to $0$ as $k$ goes to infinity.\footnote{Similarly, in the following, the $o(1)$ term is defined with respect to $k$.}\end{theorem}

It is worth noting that Theorem~\ref{thm:DenseCodeRegularLossless} is not a special case of Theorem~\ref{thm:DenseCodesPoissonActual} with $p=1$.\footnote{This arises from the fact that the latter result is based on the condition that $\gamma^{*}$ goes to $0$ (i.e., $N_T/w$ has to go to infinity) as $N_T$ goes to infinity. That is, the length of the partitions needs to go to infinity with $N_T$. However, thinking of partitions in the proof of the former result, it can be seen that each partition has length one.} In fact, Theorem~\ref{thm:DenseCodeRegularLossless} provides a tighter bound compared to the result of Theorem~\ref{thm:DenseCodesPoissonActual} with $p = 1$.

We now study the average coding delay of dense codes with respect to the traffics with deterministic regular transmission opportunities and Bernoulli losses. It should be clear that, in this case, the deviation of the number of packets per partition should not be taken into account. Thus, by replacing $r$ with $\varphi$ in Lemmas~\ref{lem:Omegaii} and~\ref{lem:Omegaij}, and redefining $w$ as $\sqrt{{p N_T L}/{\log(N_T L/\epsilon)}}$, we have the following result.\footnote{Note that the latter choice of $w$ is much larger than that in Lemma~\ref{lem:DensityBound}. This is because, in this case, there is no gap between the lower bound on the number of packet transmissions in each partition and the expectation, and hence, the partitions do not need to be sufficiently long (see Footnote~\ref{ftnt:PartitionLength}).}

\begin{theorem}\label{thm:DenseCodesPoissonAverage} The $\epsilon$-constrained average coding delay of a dense code over a network similar to Theorem~\ref{thm:DenseCodesPoissonActual} is b.a.b. \[\frac{1}{p}\left(k +(1+o(1))\left(\frac{k L}{w}+w\log\frac{wL}{\epsilon}\right)\right)\] where $w\sim \left({kL/\log(kL/p\epsilon)}\right)^{1/2}$.\end{theorem}

\begin{proof}\renewcommand{\IEEEQED}{} The proof follows the same line as that of the proof of Theorem~\ref{thm:DenseCodesPoissonActual}, except that $r$ needs to be replaced with $\varphi$ in the proof of Lemma~\ref{lem:DensityBound}. Thus, the term $O(v\sqrt{\varphi\log(w_T/\epsilon)})$ in~\eqref{eq:Temp1} and $O(\sqrt{p N_T w\log(w L/\epsilon)})$ in~\eqref{eq:Temp2} will disappear. Then, it should not be hard to see that, in this case, $w$ needs to be chosen in order to minimize $O(p N_T L/w)$ $+$ $O(w\log(w L/\epsilon))$, instead of~\eqref{eq:Temp3}. This can be done by redefining $w$ as \[\hspace{2.75 in}\sqrt{\frac{p N_T L}{\log(N_T L/\epsilon)}}. \hspace{2.75 in}\IEEEQEDclosed\]\end{proof}

\subsection{Non-Identical Links}
The preceding results regarding the identical links immediately serve as upper bounds for the case of non-identical links with parameters $\{p_i\}_{1\leq i\leq L}$, by replacing $p$ with $\min_{1\leq i\leq L}p_i$. The results however might not be very tight, e.g., for the case where, for some $1\leq i\leq L$, $p_i$ is much larger than $p$. Thus the values and the ordering of $\{p_i\}$ needs to be taken into consideration to derive tighter bounds. In particular, for every $1\leq i< L$, depending on whether the $i\textsuperscript{th}$ or the $(i+1)\textsuperscript{th}$ link has a larger parameter, Lemma~\ref{lem:VerticalT} or~\ref{lem:HorizontalT} is useful to lower bound the rank of the transfer matrix at the $i\textsuperscript{th}$ node. The rest of the analysis remains the same.

In the following, however, we state the main results for a special case of non-identical links, where there is a single worst link (a unique link with the smallest success parameter).

\begin{theorem}\label{thm:DenseCodesRegularBernoulliActualNon-Identical} Consider a sequence of parameters $\{p_i\}_{1\leq i\leq L}$ with a unique minimum $p \doteq \min_{i}p_i$. Then, the $\epsilon$-constrained coding delay of a dense code over a line network of $L$ links with deterministic regular traffics and Bernoulli losses with non-identical parameters $\{p_i\}$ is b.a.b. \[\frac{1}{p}\left(k+(1+o(1))\left(\frac{kL}{w}+\sqrt{k\left(w\log\frac{wL}{\epsilon}\right)}\right)\right)\] where $w\sim\left(k L^2/\log(k L/p\epsilon)\right)^{1/3}$.\end{theorem}

\begin{theorem}\label{thm:DenseCodesRegularBernoulliAverageNon-Identical} The $\epsilon$-constrained average coding delay of a dense code over a network similar to Theorem~\ref{thm:DenseCodesRegularBernoulliActualNon-Identical} is b.a.b. \[\frac{1}{p}\left(k +(1+o(1))\left(\frac{k L}{w}\right)\right)\] where $w\sim k/\left(p\log(k L/p\epsilon)f(k)\right)$ and $f(k)$ goes to infinity arbitrarily slow, as $k$ goes to infinity.\end{theorem}

\section{Poisson Traffic: Lossless or Bernoulli Losses}\label{sec:PoissonTraffic}
In the case of the lossless Poisson traffic with parameter $\lambda$, the number of packets in each partition of length $N_T/w$ is a Poisson random variable with the expected value $\lambda N_T/w$. By applying the Chernoff bound to the Poisson random variable (see~\cite[Theorem~A.1.15]{AS:2008}), the main results in Section~\ref{sec:BernoulliLossRegularTraffic} are applicable to this network scenario, where $p$ is replaced by $\lambda$.

In the case of Bernoulli losses over a Poisson traffic with parameters $p$ and $\lambda$, respectively, it can be shown that the points in time at which the arrivals/departures occur follow a Poisson process with parameter $\lambda p$, and hence the number of packets in each partition has a Poisson distribution with the expected value $\lambda p N_T/w$. Thus the main results in Section~\ref{sec:BernoulliLossRegularTraffic} apply by replacing $p$ with $\lambda p$.

\section{Comparison with The Existing Literature}
The upper bounds on the $\epsilon$-constrained coding delay and average coding delay, derived in this paper, are valid for any arbitrary choice of $\epsilon$. However, in the following, to compare our results with those of~\cite{PFS:2005} and~\cite{DDHE:2009}, we focus on the case where $\epsilon$ goes to $0$ polynomially fast, as $k$ goes to infinity. For such a choice of $\epsilon$, the upper bounds on the coding delay and the average coding delay hold w.p. $1$, as $k$ goes to infinity.

\subsection{Identical Links}
In~\cite{PFS:2005}, the average coding delay of dense codes over the networks of length $2$ with deterministic regular transmissions and Bernoulli losses with parameter $p$ is shown to be upper bounded by $\frac{1}{p}(k+O(\sqrt{k\log k}))$. The result of Theorem~\ref{thm:DenseCodesPoissonAverage} indicates that the average coding delay of dense codes over the networks of length $L$ with similar traffics as above is upper bounded by $\frac{1}{p}(k+(1+o(1))(\sqrt{kL\log(kL)}))$. This is consistent with the result of~\cite{PFS:2005}, although the bound presented here provides more details.

The result of Theorem~\ref{thm:DenseCodesPoissonActual} suggests that the coding delay of dense codes over network scenarios as above is upper bounded by $\frac{1}{p}(k+(1+o(1))({k^2 L \log(kL)})^{1/3})$. One should note that there has been no result on the coding delay of dense codes over identical links in the existing literature. In fact, this was posed as an open problem in \cite{DDHE:2009}. It is also noteworthy that unlike the analysis of~\cite{DDHE:2009}, our analysis does not rely on the existence of a single worst link, and hence is applicable to the case of identical links.

By combining Theorems~\ref{thm:DenseCodesPoissonActual} and~\ref{thm:DenseCodesPoissonAverage}, it can be seen that the coding delay might be much larger than the average coding delay. This highlights the fact that the analysis of the average coding delay does not provide a complete picture of the speed of convergence of dense codes to the min-cut capacity of line networks with identical links.

\subsection{Non-Identical Links}
In~\cite{DDHE:2009}, the average coding delay of dense codes over the networks of length $L$ with deterministic regular transmission opportunities and Bernoulli losses with parameters $\{p_i\}_{1\leq i\leq L}$ was upper bounded by $\frac{k}{p}+\sum_{i\neq m} \frac{1-p}{p_i-p}$, where $p = \min_i p_i$ is the unique minimum and $m = \arg\min_i p_i$. This result was derived under the unrealistic assumption that all the coded packets are innovative.

Related to this result, Theorem~\ref{thm:DenseCodesRegularBernoulliAverageNon-Identical} indicates that the average coding delay of dense codes over line networks with non-identical links is upper bounded by $\frac{1}{p}(k+(1+o(1))(p\log(kL)f(k)))$, where $f(k)$ goes to infinity arbitrarily slow, as $k$ goes to infinity. It is important to note that Theorem~\ref{thm:DenseCodesRegularBernoulliAverageNon-Identical} does not have the limiting assumption of the result of~\cite{DDHE:2009} regarding the innovation of all the packets. The bound of Theorem~\ref{thm:DenseCodesRegularBernoulliAverageNon-Identical} is larger than that of~\cite{DDHE:2009}, which is expected, since the former, unlike the latter, is derived based on the realistic assumption of operating over a finite field, which has the consequence that not all the coded packets are innovative.

The result of Theorem~\ref{thm:DenseCodesRegularBernoulliActualNon-Identical} indicates that the coding delay is upper bounded by $\frac{1}{p}(k+(1+o(1))(k^2 L\log(kL))^{1/3})$. This is while, in~\cite{DDHE:2009}, the coding delay is upper bounded by $\frac{1}{p}(k+O(k^{3/4}))$. This bound is looser than the bound in Theorem~\ref{thm:DenseCodesRegularBernoulliActualNon-Identical}, although it is derived under the same limiting assumption as the one used in~\cite{DDHE:2009} for the average coding delay (i.e., all coded packets being innovative). Such an assumption makes the bound appear smaller than what it would be at the absence of the assumption. This demonstrates the strength of the bounding technique used in this work.

Similar to the case of identical links, in the case of non-identical links, by combining Theorems~\ref{thm:DenseCodesRegularBernoulliActualNon-Identical} and~\ref{thm:DenseCodesRegularBernoulliAverageNon-Identical}, it can be seen that the coding delay might be much larger than the average coding delay. In fact, the difference might be even larger than that of the identical links.

\bibliographystyle{IEEEtran}
\bibliography{RefsII}





\end{document}